\documentclass{article}
\usepackage{latexsym,amssymb,amsmath,algo}
\usepackage{amsthm,geometry}
\usepackage{pst-node}
\usepackage{pst-coil}
\newtheorem{proposition}{Proposition}
\newtheorem{theorem}[proposition]{Theorem}
\newtheorem{corollary}[proposition]{Corollary}

\newtheorem{lemma}[proposition]{Lemma}

\newcommand{\chern}{Chernoff/\nolinebreak[0]Hoeffding}
\newcommand{\expec}[1]{\mathbb{E} \left [#1 \right ] }
\newcommand{\prob}[1]{\mathbb{P} \left [ #1 \right ] }
\newcommand{\bindist}{\mathbf{Bin}}
\newcommand{\comment}[1]{}
\newcommand{\rank}{\mathop{\mathrm{rk}}}
\newcommand{\eat}[1]{}
\newlength{\MyLength}
\title{Sorting and Selection with Random Costs}

\author{Stanislav Angelov
\thanks{Department of Computer and Information Science,
University of Pennsylvania,  {\tt angelov@cis.upenn.edu}.}
\and Keshav Kunal
\thanks{Department of Computer and Information Science,
University of Pennsylvania, {\tt kkunal@cis.upenn.edu}.}
\and Andrew McGregor
\thanks{Information Theory and Applications Center, University of California, San Diego,  {\tt andrewm@ucsd.edu}.}
    }

\begin{document}
\maketitle

\begin{abstract}
There is a growing body of work on sorting and selection in models other than
the unit-cost comparison model. This work is the first treatment of a natural
stochastic variant of the problem where the cost of comparing two elements is a
random variable. Each cost is chosen independently and is known to the
algorithm. In particular we consider the following three models: each cost is
chosen uniformly in the range $[0,1]$, each cost is 0 with some probability $p$
and 1 otherwise, or each cost is 1 with probability $p$ and infinite otherwise.
We present lower and upper bounds (optimal in most cases) for these problems.
We obtain our upper bounds by carefully designing algorithms to ensure that the
costs incurred at various stages are independent and using properties of random
partial orders when appropriate.
\end{abstract}

\section{Introduction}
\label{sec:intro}

In the relatively recent area of priced information
\cite{CL05a,CL05b,CFGKRS02}, there is a set of \emph{facts} each of which can
be \emph{revealed} at some cost. The goal is to pay the least amount such that
the revealed facts allow some inference to be made. A specific problem in this
framework, posed by Charikar et al.~\cite{CFGKRS02}, is that of sorting and
selection where each comparison has an associated cost. Here we are given a set
$V$ of $n$ elements and the cost of comparing two elements $u$ and $v$ is
$c_{(u,v)}$. This cost is known to the algorithm. We wish to design algorithms
for sorting and selection that minimize the total cost of the comparisons
performed. Results can be found in \cite{KK03,GK01,GK05} where the performance
of the algorithms is measured in terms of competitive analysis. In all cases
assumptions are made about the edge costs, e.g., that there is an underlying
monotone structure \cite{KK03,GK01} or metric structure \cite{GK05}.

A related problem that predates the study of priced information is the problem
of \emph{sorting nuts and bolts} \cite{ABFKNO94,KMS98}. This is a problem that
may be faced by  ``any disorganized carpenter who has a mixed pile of bolts and nuts and wants to find the corresponding pairs of bolts and nuts" according to
the authors of \cite{ABFKNO94}. The problem amounts to sorting two sets, $X$
and $Y$, each with $n$ elements given that comparisons are only allowed between
$u\in X$ and $v\in Y$. It can be shown that this problem can be generalized to
the priced information problem in which comparison costs are either 1 or
$\infty$.

In this paper we study a natural stochastic variant of the sorting problem. We
consider each comparison cost to be chosen independently at random.
Specifically, we consider the following three models:
\begin{enumerate}
\item[(a)] Uniform Comparison Costs: $c_{(u,v)}$ is chosen uniformly in the range $[0,1]$,
\item[(b)] Boolean Comparison Costs: $c_{(u,v)}=0$ with probability $p$ and $1$ otherwise, and
\item[(c)] Unit and Infinite Comparison Costs:  $c_{(u,v)}=1$ with probability $p$ and $\infty$ otherwise.
\end{enumerate}
This first model is in the spirit of the work on calculating the expected cost
of the minimum spanning tree \cite{F85}. The second and third models are
related to the study of random partially ordered sets (see \cite{B93} for an
overview) and linear extensions \cite{F76,KK95,ABBJ94}. Specifically, in
Model~(b), the free comparisons define a partial order $(V,\preccurlyeq)$ that
is chosen according to the \emph{random graph model}. To sort $V$ we need to do the minimum number of remaining comparisons to determine the linear extension,
or total order, defined by the remaining comparisons. In Model~(c) we have the
problem of inferring properties of the random partial order $(V,\preccurlyeq)$
defined by the cost $1$ edges.


\subsection{A Motivation from Game Theory}
\label{sec:gametheory} The framework of priced information lends itself
naturally to a game theoretic treatment where there are numerous sellers each
owning one or more facts. Some facts will be, in a sense, more valuable than
others. In the case of sorting, the value of a comparison $(u,v)$ is inversely
related to $|\{w: u<w<v \mbox{ or } v<w<u\}|$ because for each such $w$, the
comparisons $(u,w)$ and $(w,v)$ together provide an alternative way of implying
$(u,v)$. How should sellers price their information in an effort to maximize
their profit? Herein lies the dilemma --- if the pricing of the facts is
strictly monotonic with their value, the buyer can infer the sorted order from
the prices themselves and by performing a single (cheapest) comparison! Yet, if
there is no correlation, the seller is not capitalizing on the value of the
information they have to sell. It seems likely that the optimum pricing of a
fact will be a non-deterministic function of the value. While a treatment of
the game theoretic problem seems beyond our reach at this time, we feel that a
first step will be to find optimal buyer strategies when the price of each fact
is chosen randomly and independently of the value of the fact.

\subsection{Our Results}

For $p=1/2$, our results are summarized in Table~\ref{tab:results1}. In
general, we will present bounds in terms of both $n$ and $p$. Note that rather
than using a competitive analysis of our algorithms (as in
\cite{KK03,GK01,GK05}) we estimate the expected cost of our algorithms and the
expected cost of the respective minimum certificate.

\begin{table}[th]

\begin{center}
\begin{tabular}{|c|c|c|c|c|c|c|}
\hline
 & \multicolumn{2}{c|}{Max and Min} &  \multicolumn{2}{c|}{Selection} &  \multicolumn{2}{c|}{Sorting}  \\
 & \tiny{Upper Bound} & \tiny{Min. Certificate} &  \tiny{Upper Bound} & \tiny{Min. Certificate} & \tiny{Upper Bound} & \tiny{Min. Certificate}  \\ \hline
$c_{(u,v)}=1/2$ & $O(n)$ & $\Omega (n)$ & $O(n)$ & $\Omega (n)$ & $O(n\log n)$ &
$\Omega (n)$  \\\hline $c_{(u,v)}\in [0,1]$ & $O(\log n)$ & $\Omega (\log n)$ &
$O(\log^6 n)$ & $\Omega (\log n)$ &$O(n)$ & $\Omega (n)$ \\\hline $c_{(u,v)}\in
\{0,1\}$ & $O(1)$ & $\Omega (1)$ & $O(\log n)$ & $\Omega (1)$ &$O(n)$ & $\Omega
(n)$\\\hline $c_{(u,v)}\in \{1,\infty\}$ & $O(n\log n)$ & $\Omega(n)$ & $-$ & $-$ &
$-$ &$-$ \\\hline
\end{tabular}
\end{center}
\caption{Comparison between the expected costs of our algorithms and the
minimum certificates for sorting and selection for various cost functions. The
first row follows from standard algorithms and is given as a reference point
for comparison. Also, in the case of $c_{(u,v)}\in \{1,\infty\}$ we consider
finding all maximal/minimal elements.} \label{tab:results1}
\end{table}

We would like to note that for the first three rows of Table~\ref{tab:results1},
the expected cost of each comparisons is $1/2$ but the variance differs. For
selection type problems the variance makes a big difference since there are
many ways to certify the rank of an element. However for sorting there is only
one (minimal) certificate for the sorted order. Nevertheless, a little bit of
variance makes it possible to sort with only linear cost rather than $O(n\log
n)$ cost.

One of the main challenges in the analysis of our algorithms is to ensure that
the costs incurred at various stages of the algorithm are independent. We
achieve this by carefully designing the algorithms and describing an
alternative random process of cost assignment that we argue is equivalent to
the original random process of cost assignment.

\section{Preliminaries}
We are given a set $V$ of $n$ elements, drawn from some totally ordered set. We
are also given a non-negative symmetric function $c : V \times V \rightarrow
\mathcal{R}^+$ which determines the cost of comparing two elements of $V$.
Given $V$ and $c$, we are interested in designing algorithms for sorting and
selection that minimize the total cost of the performed comparisons.

The above setting is naturally described by the complete weighted graph on $V$,
call it $G$, where the weight $c_e$ of an edge $e$ is determined by the cost
function $c$. The direction of each edge $(u,v)$ in $G$ is consistent with the
underlying total order and is unknown unless the edge $e$ is \emph{probed},
i.e., the comparison between $u$ and $v$ is performed, or it is implied by
transitivity, i.e., a directed path between $u$ and $v$ is already revealed. In
this case we call $u$ and $v$ \emph{comparable}.

An algorithm for sorting or selection should reveal a \emph{certificate} of the
correctness of its output. In the case of sorting, the minimal certificate is
unique, namely the Hamiltonian path in $G$ between the largest and the smallest elements of $V$. In the case of selection, the certificate is a subgraph of $G$
that includes a (single) directed path between the element of the desired rank
and each of the remaining elements of $V$. In the special case of max-finding,
the certificate is a rooted tree on $V$, the maximum element being the root.
The \emph{cost} of a certificate is defined as the total cost of the included
edges.

In this paper we consider three different stochastic models for determining the
cost function $c$ (see Section \ref{sec:intro}). In Models (b) and (c), the
graphs induced respectively by the cost $0$ or $1$ edges have natural analogue
to random graphs with parameter $p$, denoted by $G_{n,p}$. Note that in Models
(a) and (b), the maximum cost of a comparison is $1$. When this is case, the
following proposition will be useful and follows from a natural greedy strategy
to find the maximum element in the standard comparison model.

\begin{proposition}\label{prop:findmax}
Given a set $V$ of $n$ elements, drawn from a totally ordered set, where the
cost of the comparison between any two elements is at most $1$, we can find (and
certify) the maximum element performing $n-1$ comparisons incurring a cost of
at most $n-1$.
\end{proposition}

We will measure the performance of our algorithms by comparing the expected
total cost of the edges probed with the expected cost of a minimum certificate.
Note that the cost of the minimum certificate is concentrated around the mean
in most cases. Even when the minimum certificate cost is far from the mean, we
can obtain good bounds on the expected ratio by using algorithms from
\cite{CFGKRS02} (Model (a)) or standard algorithms (Model (b) and (c)).

Finally, in the analysis, it would be often useful to number the elements of
$V$, $v_1, \cdots, v_n$ such that $v_1 < \cdots < v_n$. We also define the
\emph{rank} of an element $v$ with respect to a set $S \subseteq V$ to be
\[\rank_S(v)=|\{ u: u\leq v,u\in
S\}|\enspace .\]


\section{Uniform Comparison Costs}

In this section we will assume that the cost of each comparison is chosen
uniformly at random in the range $[0,1]$. We consider the problems of finding the maximum
or minimum elements, general selection, and sorting. The algorithms are
presented in Fig.~\ref{fig:uniformalgorithms}.

\begin{theorem}
The expected cost of \ref{alg:findmax} is at most $ 2(H_n-1)$ where
$H_k=\sum_{i=1}^k 1/i$.
\end{theorem}
\begin{proof}
We analyze a random process where we consider edges one by one in a
non-decreasing order of their cost. Note that the costs of edges define a
random permutation on the edges. If an edge is incident to two candidate
elements, i.e., elements that have not lost so far a performed comparison, we
probe the edge, otherwise we ignore the edge. Either way we say the edge is
\emph{processed}.

We divide the analysis in rounds. A round terminates when an edge is probed.
After the end of a round, the number of candidates for the maximum decreases by
one. Therefore after $n-1$ rounds the last candidate would be the maximum
element. For $r\in [n-1]$, let $t_r$ denote the random variable which counts
the number of edges processed  in the $r$th round. Let $T_r = \sum_{i = 1}^{r}
t_i $ denote the rank of the edge (in the sorted by costs order) found in the
$r$th round. Therefore, the expected cost of the performed comparison is
$\expec{T_r}/({n \choose 2} + 1)$.

It remains to show an upper bound on the value of $\expec{T_r} = \sum_{i=1}^{r}
\expec{t_i}$. So far $T_{r-1}$ edges have been processed. The probability that
the next edge is between two candidate elements is $p = {n-(r-1) \choose 2}/
\left( {n \choose 2} - T_{r-1} \right) \geq  { {n-(r-1)} \choose 2}/ {n \choose
2 } $. Hence, for $r\in [n-2]$, $\expec{t_r} \leq 1/p \leq  {n \choose 2 }/{
{n-(r-1)} \choose 2} $, and for $r=n-1$, we have $\expec{T_{r} } \leq {n\choose
2}$. We conclude that the total expected cost is at most, \[ \sum_{r=1}^{n-1}
\frac{\expec{T_r}}{({n\choose 2}+1)}  \le  1+ \sum_{r=1}^ {n-2} \sum_{i=1}^{r}
\frac1{{n-(i-1) \choose 2}} \le 1+ \sum_{r=1}^{n-2} \frac{2}{n-r+1} =  2(H_n-1)
\enspace .
\]
\end{proof}

\begin{theorem}
The expected cost of the cheapest rank $k$ certificate is $H_{k}+H_{n-k+1}-2$.
\end{theorem}
\begin{proof}
Consider $v_i$ with $i<k$. Any certificate  must include a comparison with at
least one of $v_{i+1}, \ldots, v_k$. The expected cost of the minimum of these
$k-i$ comparisons is $\frac1{k-i+1}$. Summing over $i$, $i<k$, yields
$H_{k}-1$. Similarly, now consider $v_i$ with $i>k$. Any certificate  must
include a comparison with at least one of $v_{k}, \ldots, v_{i-1}$. The
expected cost of the minimum of these $i-k$ comparisons is $\frac1{i-k+1}$.
Summing over $i$, $n\geq i>k$, yields $H_{n-k+1}-1$. The theorem follows.
\end{proof}

Note that the theorem above also implies a lower bound of $\Omega(\log n)$ on
the expected cost of the cheapest certificate for the maximum (minimum)
element. To prove a bound on the performance of \ref{alg:selection} we need the
following preliminary lemma.

\begin{lemma}\label{thm:decendents}
Let $v\in V$ and perform each comparison with probability $p$. Then, with
probability at least $1-1/n^4$ (assuming $p>1/n^3$), for all $u$ such that
\[ |\rank_V u - \rank_V v| \ge \frac{150 \log n \left(\log n +\log (1/p)\right)}{p} \enspace ,\]
 the relationship between $u$ and $v$ is certified by the comparisons performed.
\end{lemma}
\begin{proof}
\emph{W.l.o.g.}, $\rank_V v\geq n/2$. We will consider elements in $S= \{u:
\rank_V u < \rank_V v\}$. The analysis for elements among $\{u: \rank_V u >
\rank_V v\}$ is identical and the result follows by the union bound. Throughout the proof we will assume that $n$ is sufficiently large.

Let $D$ be the subset of $S$ such that $u\in D$ is comparable to $v$. We
partition $S$ into sets,
\[B_i=\{u: \rank_V v-wi \leq  \rank_V u< \rank_V v-w(i-1)\} \enspace , \]
where $w=\frac{12 \log n}{p(1-e^{-1})}$. Let $X_i= D \cap B_i$, that is, the
elements from set $B_i$ that are comparable to $v$. For the sake of
notation, let $X_0=\{v\}$. Let $D_i=\bigcup_{0\leq j\leq {i-1}} X_j$. If we
perform a comparison between an element of $D_i$ and an element $u$ of $B_i$
then we certify that $u$ is  less than $v$. The probability that an element of
$B_i$ gets compared to an element of $D_i$ is,
\[
1-(1-p)^{D_i}\geq  1-e^{-pD_i} \geq (1-e^{-1})\max \{1,pD_i\}\enspace .
\]
Let $(Y_i)_{1\leq i}$ be a family of independent random variables distributed
as $\bindist(w,q)$ where $q=(1-e^{-1})\max \{pD_i,1\}$. Note that
$\expec{Y_i}=12D_i\log n$ if $pD_i\leq 1$.
\begin{enumerate}
\item
For $i$ such that $D_i<1/p$. Using the {\chern} Bounds,
\[ \prob{X_i< D_i \log n}=\prob{X_i<\frac{qw}{12}}\leq \prob{Y_i<\frac{\expec{Y}}{12}}\leq e^{-6 (11/12)^2 D_i\log n}\leq \frac{1}{n^5} \enspace .\]
In other words, the number of comparable elements increases by at least a $\log
n$ factor until $D_i\geq 1/p$. Therefore, with probability at least
$1-\log(1/p)/n^5$, for all $i$,
\[D_i\geq\min \{ 1/p, (\log n)^{i-1}\}  \enspace .\]
In particular, $D_{\log 1/p} \geq 1/p$.
\item
Assume that  $i>\log (1/p)$ and therefore  $D_{\log 1/p} \geq 1/p$. Using
{\chern} Bounds, we get
\[ \prob{X_i<\frac1p}\leq \prob{Y_i<\frac1p}\leq \prob{Y_i< \frac{\expec{Y_i}}{12\log n}}\leq e^{-(1-\frac1{12 \log n})^2 6 D_i \log n }\leq \frac1{n^5}\enspace .\]
\end{enumerate}
Therefore with probability at least $1-\frac{t}{n^5}$, $D_{t}\geq \frac{6\log
n}{p}$ where $ t=6\log n+\log (1/p)$.
 Consider an element
$u \in B_{t'}$ where $t' > t$.  The probability that $u$ is not in $D$ is
bounded above by $(1-p)^{6\log n /p}\leq 1/n^6$. Hence, with probability at
least $1-(1+t)/n^5$,
\[|S\setminus D| \leq wt\leq \frac{150 \log n \left(\log n+ \log
(1/p)\right)}{p}\enspace .\]
\end{proof}

\begin{figure}[t]
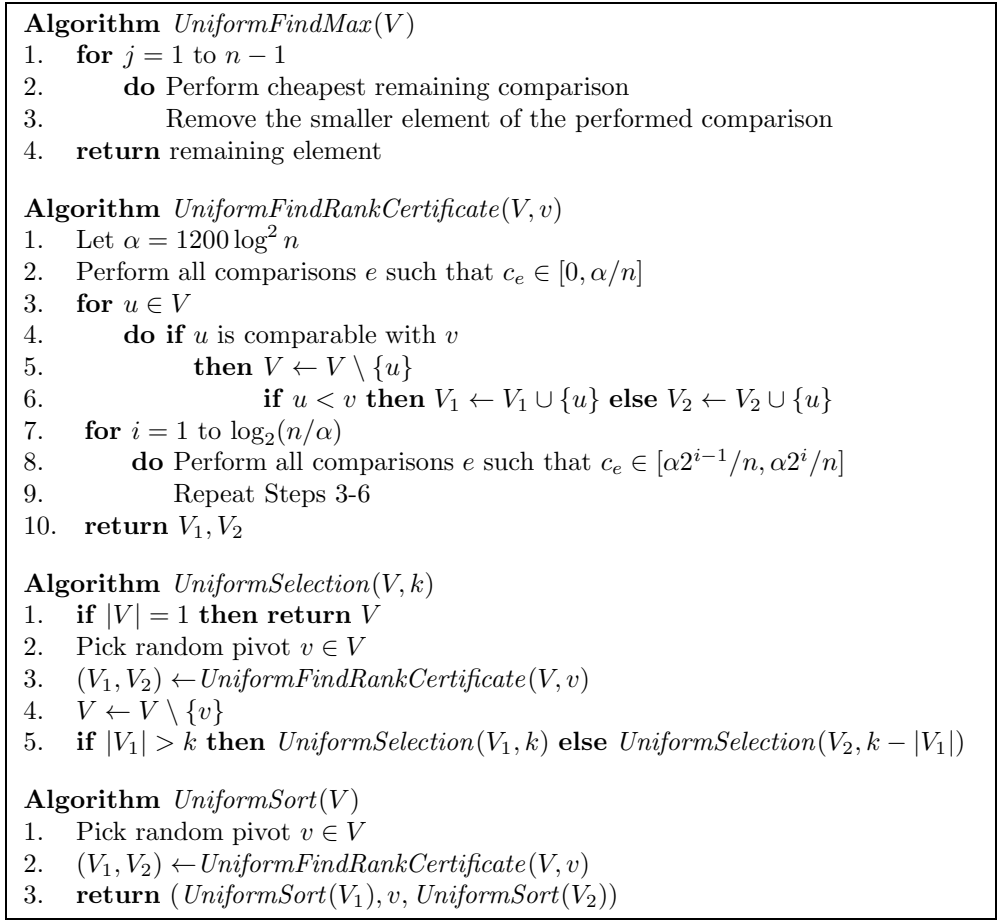

\begin{center}
\fbox{
\begin{minipage}{5in}
{
\begin{algorithm}{UniformFindMax}[V]{
\label{alg:findmax}}
\qfor $j=1$ to $n-1$\\
\qdo
Perform cheapest remaining comparison\\
Remove the smaller element of the performed comparison
\qrof \\
\qreturn remaining element
\end{algorithm}

\begin{algorithm}{UniformFindRankCertificate}[V,v]{
\label{alg:pivotcert}}
Let $\alpha=1200\log^2 n $\\
Perform all comparisons $e$ such that $c_e\in [0,\alpha / n ]$\\
\qfor $u \in V$ \\
\qdo
\qif $u$ is comparable with $v$ \\
\qthen $V\leftarrow V\setminus \{u\}$\\
\qqif $u < v$ \qqthen $V_1\leftarrow V_1\cup \{u\}$ \qqelse $V_2\leftarrow V_2\cup \{u\}$
\qrof
\qrof \\
\qfor $i=1$ to $\log_2 (n/\alpha)$\\
\qdo Perform all comparisons $e$ such that $c_e\in [\alpha 2^{i-1}/ n ,\alpha 2^{i}/ n ]$\\
Repeat Steps 3-6 \qrof \\
\qreturn $V_1, V_2$
\end{algorithm}

\begin{algorithm}{UniformSelection}[V,k]{
\label{alg:selection}}
\qqif $|V|=1$ \qqthen \qreturn $V$\\
Pick random pivot $v\in V$\\
$(V_1,V_2)\leftarrow $\ref{alg:pivotcert}$(V,v)$\\
$V\leftarrow V\setminus \{v\}$\\
\qqif $|V_1|>k$ \qqthen $\ref{alg:selection}(V_1,k)$ \qqelse
$\ref{alg:selection}(V_2,k-|V_1|)$
\end{algorithm}

\begin{algorithm}{UniformSort}[V]{ \label{alg:sort}}
Pick random pivot $v\in V$\\
$(V_1,V_2)\leftarrow $\ref{alg:pivotcert}$(V,v)$\\
\qreturn $\left(\ref{alg:sort}(V_1),v,\ref{alg:sort}(V_2)\right)$
\end{algorithm}
}
\end{minipage}
} \caption{ Algorithms for uniform comparison costs.
\label{fig:uniformalgorithms}}
\end{center}
\end{figure}

\begin{lemma}\label{lem:pivotcert}
Consider the algorithm \ref{alg:pivotcert} called on a randomly chosen $v$.
With probability at least $1-n^{-3}$ the algorithm returns a certificate for
the rank of $v$. The expected cost of the comparisons is $O(\log^5 n)$.
\end{lemma}
\begin{proof}
Let $V_i$ be the set of elements at the start of iteration $i$. Let
$p_1=\alpha/n$ be the probability that $c_e\in  [0 ,\alpha / n ]$. For $i>1$,
let $p_i=\alpha 2^{i-1}/n$ be the probability that $c_e\in  [\alpha 2^{i-1}/ n
,\alpha 2^{i}/ n ]$. First we show that, with probability at least $1-
\frac{\log_2 (n/\alpha)}{n^4} $, for all $1\leq i \leq \log_2 (n/\alpha)$,
$|V_i|<n/2^{i-1}$. Assume that $|V_i|<n/2^{i-1}$. Appealing to Lemma
\ref{thm:decendents}, there are less than
\[
\frac{300\log |V_i| (\log |V_i|+\log(1/p))}{p}\leq \frac{600\log^2 n }{\alpha
2^{i-1}/n} = |V_i|/2\] elements in $V_{i+1}\setminus V_i$ and hence the
$|V_{i+1}|< n/2^{i}$. It remains to show that the cost per iteration is
$O(\log^4 n)$. This follows since the expected number of comparisons is
$O(V_{i}^2 \alpha 2^{i}/n)=O(\alpha n/2^i)$ and each comparison costs at most
$\alpha 2^{i}/ n$.
\end{proof}

The following theorem can be proved using standard analysis of the appropriate recurrence relations and Lemma \ref{lem:pivotcert}.

\begin{theorem}
The algorithm \ref{alg:selection} can be used to select the $k$\emph{th}
element. The expected cost of the  certificate is $O(\log^6 n)$. The algorithm
\ref{alg:sort}  returns a sorting certificate with expected cost $O(n)$.
\end{theorem}

Note that we can check if a certificate is a valid one without performing any additional comparisons. In the case when \ref{alg:pivotcert} fails, we can reveal all edges to obtain a certificate without increasing asymptotically the overall expected cost.

\begin{theorem}\label{thm:uniformcostcert}
The expected cost of the cheapest sorting certificate is $(n-1)/2$.
\end{theorem}
\begin{proof}
For each $1\leq i\leq n-1$ there must be a comparison between $v_i$ and
$v_{i+1}$. The expected cost of each is $1/2$. The theorem follows by linearity
of expectation.
\end{proof}

\section{Boolean Comparison Costs}

In this section we assume that comparisons are for free with probability $p$
and have cost $1$ otherwise. We consider the problems of finding the maximum or minimum
elements, general selection, and sorting. The algorithms for maximum finding
and selection are presented in Fig.~\ref{fig:01algorithms}. For sorting we use
results from \cite{ABBJ94} and \cite{KK03} to obtain a bound on the number of
comparisons needed to sort the random partial order defined by the free
comparisons.

\begin{figure}[t]
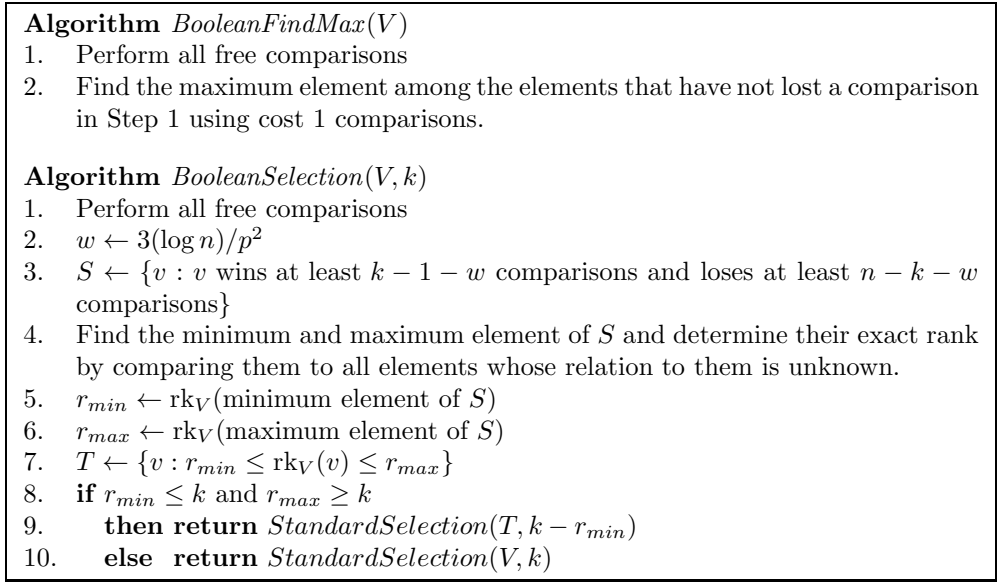

\begin{center}
\fbox{
\begin{minipage}{5in}
{
\begin{algorithm}{BooleanFindMax}[V]
{\label{alg:findmax01}}
Perform all free comparisons\\
Find the maximum element among the elements that have not lost a comparison in
Step 1 using cost 1 comparisons.
\end{algorithm}

\begin{algorithm}{BooleanSelection}[V,k]{
\label{alg:selection01}}
Perform all free comparisons\\
$w\leftarrow 3(\log n)/p^2$\\
$S \leftarrow \{v: v \textrm{ wins at least } k-1-w$ comparisons and loses at least $n-k-w$ comparisons$\}$\\
Find the minimum and maximum element of $S$ and determine their exact rank by comparing them to all elements whose relation to them is unknown. \\
$r_{min} \leftarrow \rank_V(\textrm{minimum element of } S)$\\
$r_{max} \leftarrow \rank_V(\textrm{maximum element of } S)$\\
$T \leftarrow \{v:  r_{min} \leq \rank_V(v)  \leq r_{max} \} $ \\
\qif $r_{min} \leq k$ and $r_{max} \geq k$ \\
\qthen \qreturn $StandardSelection(T, k-r_{min})$ \\
\qelse \qreturn $StandardSelection(V, k)$
\end{algorithm}
%
}
\end{minipage}
} \caption{Algorithms for boolean comparison costs \label{fig:01algorithms}}
\end{center}
\end{figure}

\begin{theorem}\label{thm:boolmaxsort}
The expected cost of \ref{alg:findmax01} is $1/p-1$ as $n\rightarrow \infty$.
\end{theorem}
\begin{proof}
Consider the $i$th \emph{largest} element. The probability that there is no free
comparison to a larger element is $(1-p)^{i-1}$. Hence, after performing all
the free comparisons, the expected number of non-losers, in the limit as $n$ tends to infinity, is 
\[\lim_{n\rightarrow \infty} \sum_{i=1}^{n}(1-p)^{i-1}=\lim_{n\rightarrow \infty}  \frac{1 - (1-p)^n}{p} = 1/p   \enspace .\]
 Hence, by Proposition \ref{prop:findmax}, the expected number of
comparisons of cost 1 that are necessary is $1/p-1$.
\end{proof}
The theorem above leads to an immediate corollary:
\begin{corollary}
The expected cost of the cheapest certificate for the maximum element and the
element of rank $k$ is $\Omega (1/p)$  as $n\rightarrow \infty$.
\end{corollary}

Using Theorem~\ref{thm:boolmaxsort}, we obtain a sorting algorithm with
expected cost of at most $(1/p-1)(n-1)$ by repeating $n-1$ times
\ref{alg:findmax01}. We improve this result (for sufficiently small $p$) by
observing that the free comparisons define a random partial order on the $n$
elements, call it $G_{n,p}$. In \cite{ABBJ94}, the expected number of linear
extensions of $G_{n,p}$ was shown to be 
\[\prod_{k=1}^n \frac{1-(1-p)^k}{p} \le
\frac{1}{p^{n-1}} \enspace .\]
 A conjecture, independently proposed by
Kislitsyn~\cite{Kislitsyn68}, Fredman~\cite{F76}, and Linial~\cite{Linial84},
states that given a partial order $P$, there is a comparison between two
elements such that the fraction of extensions of $P$ where the first elements
precedes the second one is between $1/3$ and $2/3$. Ignoring running time, this
would imply sorting with cost $\log_{3/2} e(P)$, where $e(P)$ denotes the
number of linear extensions of $P$. In \cite{KS84}, a weaker version of the
conjecture was shown giving rise to an efficient, via randomization
\cite{DyerFK89}, sorting algorithm with cost $\log_{11/8} e(P)$. Taking a
different approach, Kahn and Kim~\cite{KK95} described a deterministic
polynomial time, $O(\log e(P))$ cost algorithm to sort any partial order $P$.

Combining the above results, and using Jensen's inequality, we obtain a sorting
algorithm with expected cost at most,
$$\log_{\frac{11}{8}} e(G_{n,p})
\le (\log_{\frac{11}{8}} p^{-1}) (n-1) \enspace .$$ Note that for
$p<0.1389$, $\log_{11/8}(1/p) < 1/p - 1$. Combining the two sorting methods, we
obtain the following theorem.

\begin{theorem}
There is a sorting algorithm for the Boolean Comparison Model with expected
cost of $\min\{\log_{\frac{11}{8}}1/p,1/p-1\} \cdot (n-1)$.
\end{theorem}

The proof of the following theorem about the cheapest sorting certificate is
nearly identical to that of Theorem \ref{thm:uniformcostcert}.

\begin{theorem}
The expected cost of the cheapest sorting certificate is $(1-p)(n-1)$.
\end{theorem}

We next present our results for selection.

\begin{theorem}\label{thm:p^-2}
The algorithm \ref{alg:selection01} can be used to select the $k$\emph{th}
element. The expected cost of the algorithm is $O(p^{-2}\log n)$.
\end{theorem}
\begin{proof}
We want to bound the size of set $S$ as defined in the algorithm. Fix an
element $v_j$. For an element $v_i$ such that $i < j$, let $l=j-i-1$. Consider
the event that we can infer $v_i<v_j$ from the free comparisons because there
exists an element $v_{i'}$ such that $v_i<v_{i'}<v_k$ and
$c_{(v_i,v_{i'})}=c_{(v_{i'},v_j)}=0$. The probability of this event is
$1-(1-p^2)^l$ and hence with probability at least $1-1/n^3$ we learn $v_i<v_j$
if $l \geq w=3(\log n)/p^2$. Therefore, with probability at least $1-1/n^2$,
$v_j$ wins at least $j-1-w$ comparisons. Similarly with probability at least
$1-1/n^2$, $v_j$ loses at least $n-j+w$ comparisons.

Hence, with probability $1-2/n^2$, every element from the set \[ S' =\{ v: k-w
\leq  \rank_V v \leq k+w \} \enspace , \] belongs to the set $S$ and in
particular the element of rank $k$ also belongs to $S$. Note that no element
from outside $S'$ can belong to $S$ and hence $|S|\leq 2w$. By Proposition
\ref{prop:findmax}, it takes $O(w)$ comparisons to compute the minimum and
maximum elements in $S$. There are at most $2w$ elements incomparable to the
minimum (maximum) element with probability at least $1-2/n^2$ and hence the
expected cost for determining the exact rank of minimum (maximum) element from
$S$ is bounded by \[2w(1-2/ n^2) + (n-1)2/n^2 \,=\, O(w)\] in expectation.  Since
the size of $T$ is also $O(w)$, step $5$ takes $O(w)$ time if $v_k \in T$,
which happens with probability at least $1-2/n^2$, and $O(n)$ otherwise.
Similar to the previous step, the expected cost is $O(w)$.
\end{proof}

Note that with a slight alteration to the \ref{alg:selection01} algorithm it is
possible to improve upon Theorem~\ref{thm:p^-2} if $p $ is much smaller than
$1/\log n$. Namely, setting \[w= 150p^{-1} \log n \log (n/p)\enspace ,\] and
appealing to Lemma~\ref{thm:decendents} in the analysis, gives an expected cost
of $$O\left( p^{-1} \log n \log (n/p)\right) \enspace .$$

\section{Unit and Infinite Comparison Costs}

In this section we consider the setting where only a subset of the comparisons
is allowed. More specifically, each comparison is allowed with probability $p$
(has cost 1) and is not allowed otherwise (has infinite cost). Here, the
underlying total order might not be possible to infer even if all comparisons
are performed. This is because, for example, adjacent elements can be compared
only with probability $p$. Hence, even the maximum element might not be
possible to certify exactly. We therefore relax our goals to finding maximal
elements and inferring the poset defined by the edges of cost $1$. In what
follows, we present algorithms for finding a maximal element as well as all
maximal elements (see Fig.~\ref{fig:partordalgorithms}). We consider an element
maximal if it wins (directly or indirectly) all allowed comparisons to its
neighbors.

\begin{figure}[t]
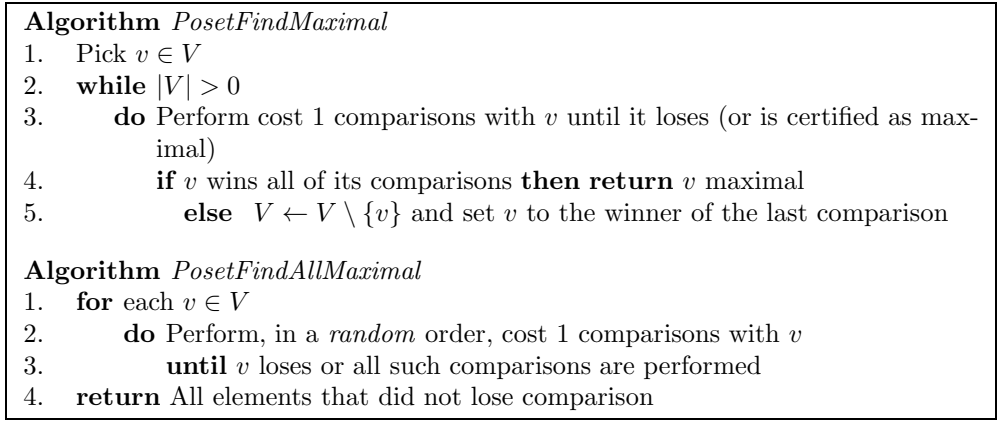

\begin{center}
\fbox{
\begin{minipage}{5in}
{
\begin{algorithm}{PosetFindMaximal}{
\label{alg:findmax1inf}}
Pick $v \in V$\\
\qwhile $|V| > 0$\\
\qdo Perform cost 1 comparisons with $v$ until it loses (or is certified as maximal)\\
\qif $v$ wins all of its comparisons \qqthen \qreturn $v$ maximal\\
\qelse $V \leftarrow V \setminus \{v\}$ and set $v$ to the winner of the last
comparison
\end{algorithm}

\begin{algorithm}{PosetFindAllMaximal}{
\label{alg:findallmax1inf}}
\qfor each $v\in V$\\
\qdo Perform, in a \emph{random} order, cost 1 comparisons with $v$\\
\textbf{until} $v$
loses or all such comparisons are performed \qrof \\
\qreturn All elements that did not lose comparison
\end{algorithm}
}
\end{minipage}
} \caption{Algorithms for $1/\infty$ comparison costs
\label{fig:partordalgorithms}}
\end{center}
\end{figure}

\begin{theorem}
The expected cost of the cheapest certificate for all maximal elements is
\[\Omega\left(n(1 - (1-p)^{n-1})\right)\enspace .\]
\end{theorem}
\begin{proof}
In this setting, each element that has no edges of cost 1 incident to it is a
maximal element. In expectation, there are $n(1-p)^{n-1}$ such elements. For
each of the remaining elements we need to do at least one comparison. Note that each comparison satisfies this requirement for two elements.
Therefore, we need to do at least $\frac{1}{2}(n - n(1-p)^{n-1})$ comparisons
in expectation. 
\end{proof}

\begin{theorem}
The expected cost of \ref{alg:findallmax1inf} is $O(n \log n)$. The expected
cost of \ref{alg:findmax1inf} is at most $n-1$.
\end{theorem}
\begin{proof}
We first analyze \ref{alg:findallmax1inf}. Fix an element $v$. Let $i =
\rank_V(v)$. Consider the following equivalent random process that assigns costs ($1$ or $\infty$)
to edges in the following way:
\begin{enumerate}
\item Pick $t$ from a random variable $T$ distributed as $\bindist(n-1,p)$.
\item Repeat $t$ times: Assign cost 1 to a random edge adjacent to $v$ whose cost has not yet been determined.
\item Declare the cost of all other edges adjacent to $v$ to be $\infty$.
\item For each remaining graph edge assign cost 1 with probability $p$ and $\infty$ otherwise.
\end{enumerate}
We may assume that the algorithm probes the cost 1 edges in this order until
$v$ loses a comparison or until all cost 1 edges are revealed. If $v$ has not
lost a comparison, there is a probability of at least $(i-1)/(n-1)$ that $v$
loses the next performed comparison. Hence, the expected number of comparisons
involving $v$ is
\[ \sum_t \prob{T=t} \sum_{j=1}^{t}\frac{i-1}{n-1} \left ( 1-\frac{i-1}{n-1}\right )^{j-1} j
\leq\sum_t \prob{T=t} \frac{n-1}{i-1} \leq \frac{n-1}{i-1} \enspace .\]
Therefore, by linearity of expectation the total number of comparisons we
expect to do is bounded above by $(n-1)H_{n-1}+(n-1)$.

The second part of the theorem follows easily from Proposition 1. The algorithm
\ref{alg:findmax1inf} is given for completeness.
\end{proof}

\section{Conclusions and Open Questions}

We have presented a range of algorithms for finding cheap sorting/selection
certificates in three different stochastic priced-information models. Most of
our algorithms are optimal up to constants and the remaining algorithms are
optimal up to poly-logarithmic terms (for constant values of the parameter $p$). Beyond improving the
existing algorithms there are numerous ways to extend this work. In particular,
\begin{itemize}
\item What about the price model in which the comparison costs are chosen in
an adversarial manner but the order of the elements is randomized?

\item In this work we have compared expected cost of minimum certificates to expected cost of the algorithms presented.
Is it possible to design algorithm which are optimal in the sense that
the expected cost of the certificate found is minimal over all
algorithms? Perhaps this would admit an information theoretic
approach.
\end{itemize}

Finally, this work was partially motivated by the game theoretic framework  described
in Section~\ref{sec:gametheory}. A full treatment of this problem was beyond
the scope of the present work. However, the problem seems natural and
deserving of further investigation.

{

\bibliographystyle{abbrv}
\bibliography{sorting}

\begin{thebibliography}{10}

\bibitem{ABFKNO94}
N.~Alon, M.~Blum, A.~Fiat, S.~Kannan, M.~Naor, and R.~Ostrovsky.
\newblock Matching nuts and bolts.
\newblock In {\em SODA}, pages 690--696, 1994.

\bibitem{ABBJ94}
N.~Alon, B.~Bollob\'as, G.~Brightwell, and S.~Janson.
\newblock Linear extensions of a random partial order.
\newblock {\em Annals of Applied Probability}, 4:108--123, 1994.

\bibitem{B93}
G.~Brightwell.
\newblock Models of random partial orders.
\newblock pages 53--83, 1993.

\bibitem{CFGKRS02}
M.~Charikar, R.~Fagin, V.~Guruswami, J.~M. Kleinberg, P.~Raghavan, and
  A.~Sahai.
\newblock Query strategies for priced information.
\newblock {\em J. Comput. Syst. Sci.}, 64(4):785--819, 2002.

\bibitem{CL05a}
F.~Cicalese and E.~S. Laber.
\newblock A new strategy for querying priced information.
\newblock In H.~N. Gabow and R.~Fagin, editors, {\em STOC}, pages 674--683.
  ACM, 2005.

\bibitem{CL05b}
F.~Cicalese and E.~S. Laber.
\newblock An optimal algorithm for querying priced information: Monotone
  boolean functions and game trees.
\newblock In G.~S. Brodal and S.~Leonardi, editors, {\em ESA}, volume 3669 of
  {\em Lecture Notes in Computer Science}, pages 664--676. Springer, 2005.

\bibitem{DyerFK89}
M.~E. Dyer, A.~M. Frieze, and R.~Kannan.
\newblock A random polynomial time algorithm for approximating the volume of
  convex bodies.
\newblock In {\em STOC}, pages 375--381. ACM, 1989.

\bibitem{F76}
M.~L. Fredman.
\newblock How good is the information theory bound in sorting?
\newblock {\em Theor. Comput. Sci.}, 1(4):355--361, 1976.

\bibitem{F85}
A.~M. Frieze.
\newblock Value of a random minimum spanning tree problem.
\newblock {\em J. Algorithms}, 10(1):47--56, 1985.

\bibitem{GK01}
A.~Gupta and A.~Kumar.
\newblock Sorting and selection with structured costs.
\newblock In {\em FOCS}, pages 416--425, 2001.

\bibitem{GK05}
A.~Gupta and A.~Kumar.
\newblock Where's the winner? {M}ax-finding and sorting with metric costs.
\newblock In C.~Chekuri, K.~Jansen, J.~D.~P. Rolim, and L.~Trevisan, editors,
  {\em APPROX-RANDOM}, volume 3624 of {\em Lecture Notes in Computer Science},
  pages 74--85. Springer, 2005.

\bibitem{KK95}
J.~Kahn and J.~H. Kim.
\newblock Entropy and sorting.
\newblock {\em J. Comput. Syst. Sci.}, 51(3):390--399, 1995.

\bibitem{KS84}
J.~Kahn and M.~Saks.
\newblock Balancing poset extensions.
\newblock {\em Order}, 1:113--126, 1984.

\bibitem{KK03}
S.~Kannan and S.~Khanna.
\newblock Selection with monotone comparison cost.
\newblock In {\em SODA}, pages 10--17, 2003.

\bibitem{Kislitsyn68}
S.~S. Kislitsyn.
\newblock A finite partially ordered set and its corresponding set of
  permutations.
\newblock {\em Mathematical Notes}, 4:798 -- 801, 1968.

\bibitem{KMS98}
J.~Koml{\'o}s, Y.~Ma, and E.~Szemer{\'e}di.
\newblock Matching nuts and bolts in ${O}(n \log n)$ time.
\newblock {\em SIAM J. Discrete Math.}, 11(3):347--372, 1998.

\bibitem{Linial84}
N.~Linial.
\newblock The information-theoretic bound is good for merging.
\newblock {\em SIAM J. Comput.}, 13(4):795--801, 1984.

\end{thebibliography}
\end{document}

}